\documentclass[submission,copyright,creativecommons]{eptcs}
\providecommand{\event}{TARK 2019} 
\usepackage{breakurl}             
\usepackage{underscore}           

\title{Reasoning about Social Choice and Games in Monadic Fixed-Point Logic}

\author{Ramit Das
\institute{IMSc (HBNI), Chennai, India}
\email{ramitd@imsc.res.in}
\and
R. Ramanujam
\institute{IMSc (HBNI), Chennai, India}
\email{jam@imsc.res.in}
\and
Sunil Simon
\institute{Department of CSE \\IIT Kanpur, Kanpur, India}
\email{simon@cse.iitk.ac.in}
}

\def\titlerunning{Reasoning about Social Choice and Games in Monadic Fixed-Point Logic}
\def\authorrunning{R. Das, R. Ramanujam \& S. Simon}

\usepackage[ruled,vlined,linesnumbered,lined]{algorithm2e}

\usepackage{amsmath}
\usepackage{amsthm}
\usepackage{amssymb}
\usepackage{thmtools,thm-restate}
 \usepackage{enumerate} 
 \usepackage[many]{tcolorbox}
\usepackage{todonotes}
\usepackage{framed}
\usepackage{amsmath}

\newcommand{\lfp}{\mathrm{\bf{lfp}}}
\newcommand{\fv}{\mathit{fv}}

 \newtheorem{theorem}{Theorem}

\begin{document}
\maketitle

\begin{abstract}
Whether it be in normal form games, or in fair allocations, or in
voter preferences in voting systems, a certain pattern of reasoning
is common. From a particular profile, an agent or a group of agents
may have an incentive to shift to a new one. This induces a natural
graph structure that we call the {\em improvement graph} on the
strategy space of these systems. We suggest that the monadic
fixed-point logic with counting, an extension of monadic first-order
logic on graphs with fixed-point and counting quantifiers, is a
natural specification language on improvement graphs, and thus for
a class of properties that can be interpreted across these domains.
The logic has an efficient model checking algorithm (in the size of the
improvement graph).
\end{abstract}


%
\section{Introduction}
A logical study of game theory aims at exposing the assumptions
and reasoning that underlies the basic concepts of game theory.
This involves the study of individual, rational, strategic 
decision making between presented alternatives (in the 
non-cooperative setting). One potential form of reasoning
in such a situation is to envisage all possible strategic choices
by others, consider one's own response to each, then others' 
response to it in their turn, and so on {\em ad infinitum}, with
Nash equilibrium representing fixed-points of such iteration..

Such reasoning, which we might call \textit{improvement dynamics},
 is similar to but distinct from rational decision making under 
uncertainty; it is also similar to but distinct from epistemic 
reasoning. The former is about optimization, selecting the
`best' option in light of one's information; the latter is about
`higher order information' involving information about others'
information etc. Improvement dynamics intends to yield the same
end results as these, but operates at a more operational, 
computational level, and reasoning about it can be seen as
reasoning at the level of computations searching for equilibria.
In this sense, logic is seen as a succinct language for
describing computational structure, rather than as a deductive
system of reasoning by agents. In spirit, the role of such logics
is similar to that of logics in descriptive complexity theory.
If we were to talk of the descriptive complexity of game theoretic
equilibrium notions, it would need to account for the implicit
improvement dynamics embedded in the solution concept.

Interestingly, several contexts in social choice theory embed such
improvement dynamics as well. When we aggregate individual choices
or preferences into social choices / preferences,  or decide on
social action (like resource allocation) based in individual
preferences, once again we see implicit improvement dynamics.
If a particular profile of voter preferences yields a specific
electoral outcome, one can consider a voter announcing a revised
(and altered) preference to force a different outcome. Two agents
might exchange their allocated goods to move to a new allocation,
if they perceive advantage in doing so. Again, these can be seen
as offers and counter-offers, perhaps leading to an equilibrium,
or not.  Some of these situations involve individual improvements,
some (like pairs of agents swapping goods) involve coalitions,
but they have the same underlying computational structure.

In this paper, we suggest that \textit{monadic fixed-point logic (with 
counting)} is a suitable language for reasoning about this 
computational structure underlying games and social choice contexts.
This is an extension of first order logic with monadic least fixed-point 
operators and counting. In this, we follow the spirit of descriptive
complexity, where extensions of first order logics describe complexity 
classes. Formulas offer concise descriptions of reasoning embedded
in improvement dynamics.

Why bother? When we have a common language across contexts, we can 
employ a form of reasoning common in one (say normal form games) in
another (say fair resource allocations) and thus transfer results and
techniques. We show that the idea of improvement under swaps corresponds
to certain form of strong equilibria and coalitional improvement in games.
Dynamics in iterated voting again correponds to improvement dynamics in
games. In such cases when the structures studied possess interesting
properties such as the {\em finite improvement property} or {\em weak
acyclicity} we get certificates of existence of equilibria. Interesting
subclasses of games (such as {\em potential games}) possess such 
properties and by ``transfer'' we can look for similar subclasses in
social choice contexts, and {\em vice versa}.

The choice of monadic fixed-point logic  is also motivated by the fact
that it admits an efficient model checking algorithm. Monadic least
fixed point operator, iterating over subsets of strategy profiles,
suffices for improvement dynamics. Counting can help us constrain paths
succinctly: though counting is first order expressible, such expression 
would be prohibitively long. 

Thus the contribution of this paper is modest and simple. The reasoning
discussed is familiar, that of improvement dynamics in normal form games,
and expressing this in monadic fixed-point logic with counting. In the
process, we can study the same properties in different contexts, such
as normal form games, fair allocations and electoral systems. We also
present a model checking algorithm for the logic.

\medskip

\noindent {\bf Logic and game theory.}  Various logical formalisms
have been used in the literature to reason about games and
strategies. Action indexed modal logics have often been used to
analyse finite extensive form games where the game representation is
interpreted as models of the logical language
\cite{BonBI,Ben01,BenProcess}. A dynamic logic framework can then be
used to describe games and strategies in a compositional manner
\cite{Par85,Gor03,kr08} and encode existence of equilibrium strategies
\cite{HHMW03}. Alternating temporal logic (ATL) \cite{atl} and its
variants \cite{strat-reasoning,atles,SL10} constitute a popular
framework to reason about strategic ability in games, especially
infinite game structure defined by unfoldings of finite graphs. These
formalism are useful to analyse strategic ability in terms of
existence of strategies satisfying certain properties (for example,
winning strategies and equilibrium strategies). Some of the above
logical formalism are also able to make assertions about partial
specifications that strategies have to conform to in order to
constitute a stable outcome. In this work we suggest a framework to
reason about the dynamics involved in iteratively updating strategies
and to analyse the resulting convergence properties. \cite{BG10}
consider dynamics in reasoning about games in the same spirit as ours
and describe it in fixed-point logic. But crucially, the dynamics is
on iterated announcements of players' rationality, and belief revision
in response to it. Moreover, they discuss extensive form games rather
than normal form games.  However, they do advocate the use of the
fixed-point extension of first order logic for reasoning about games.

Monadic least fixed point logic (MLFP) is an extension of first-order
logic which is well studied in finite model theory \cite{Sch06}. It is
a restriction of first order logic with least fixed point in which
only unary relation variables are allowed. MLFP is an expressive logic
for which, on finite relational structures, model checking can be
solved efficiently \cite{EF99}. It is also known that MLFP is
expressive enough to describe various interesting properties of
games on finite graphs. MLFP can also naturally
describe transitive closure of a binary relation which makes it an
ideal logical framework to analyse the dynamics involved in updating
strategies and its convergence properties. When $\alpha$ is a formula
with one first order free variable, $C_x \alpha \leq k$ asserts that
the number of elements in the domain satisfying $\alpha$ is at most
$k$. Clearly, this is expressible in first order logic with equality,
but at the expense of succinctness. In the literature on first order
logic with arithmetical predicates \cite{Sch05}, it is customary to
consider two sorted structures to distinguish between domain elements
and the counts, but since our domain elements are always profiles,
there is no need for such caution.

It well known that a variety of contexts in the mathematical social
sciences can be formulated in terms of improvement dynamics leading
to equilibria (of some kind). Our observation here is that the deployment
of the MLFPC logic can help to unify algorithmic techniques across these
contexts. Rather than devise an algorithm for each problem of this kind,
definability in MLFPC can at once give a uniform algorithm, which could
then be fine-tuned. Admittedly when we present contexts as diverse as
normal form games, allocations in social choice theory or voting rules,
all in one uniform framework, we only get a broad-strokes description of
the models, and the literature on these contexts vary widely in details.
We hope to convince the reader that {\em a priori}, the MLFPC has 
sufficient expressiveness to capture interesting variations. Our hope
is to delineate the logical resources needed to express the variations,
but that will require more work ahead.

\section{The improvement graph structure}
Improvement dynamics is a natural notion to study in the context of
any situation involving strategic interaction of agents. In this
section we formalise this dynamics in terms of the data structure
called improvement graphs. We consider three specific application
domains: strategic form games, voting theory and allocation of
indivisible items. We show how improvement graphs can be interpreted
in these applications and argue that the analysis of the structure
acts as the basis for reasoning about strategic interaction.
 
Let $[n] = \{1,\ldots,n\}$ denote the set of $n$ agents. Each agent is
associated with a finite set of choices $S_i$. A profile of choices
(one for each agent) induces an outcome in the strategic
interaction. Let $S$ denote the set of all choice profiles, $O$ denote
the set of all outcomes and $s(O)$ denote the outcome associated with
the profile $s \in S$. Each agent $i \in [n]$ is associated with a
preference ordering over the outcome set: $\preceq_i \subseteq (O
\times O)$. This ordering induces a preference ordering over profiles
as follows: for $s, s' \in S$ and $i \in [n]$, $s \preceq_i s'$ if
$s(O) \preceq_i s'(O)$. For a choice profile $s=(s_1,\ldots,s_n)$, we
use the standard notation $s_{-i}$ to denote the $n-1$ tuple arising
from $s$ in which the choice of agent $i$ is removed.

The associated \emph{improvement graph} is the directed graph
$G=(V,E)$ where $V = S$ and $E \subseteq V \times [n] \times V$. We
will denote the triple $(s,i,s') \in E$ by $s \to_i s'$. The edge
relation $E$ satisfies the condition: for $i \in [n]$ and $s, s'\in
S$, we have $s\to_i s'$ if $s \prec_i s'$, $s_i \neq s_i'$ and $s_{-i}
= s'_{-i}$. An \emph{improvement path} in $G$ is a maximal sequence of
profiles $s^1 s^2 \cdots$ such that for every $j > 0$ there is a
player $k_j$ such that $s^j \rightarrow_{k_j} s^{j+1}$.  Note that
here we use deviation by a single player to define the improvement
graph. We could easily extend the definition to deviation by a subset
of players, this interpretation might be more relevant in certain
domains.

\subsection{Strategic form games}
A strategic form game is given by the tuple $T=([n],\{S_i\}_{i \in
  [n]}, O, \lambda,\{\preceq_i\}_{i \in [n]})$ where the set of
strategies $S_i$ for agent $i \in [n]$ can be viewed as its set of
choices. For $S=S_1 \times \cdots \times S_n$, the function $\lambda:
S \to O$ associates an outcome to every strategy profile. 
In this paper, we consider only {\em finite} strategic form games.
The notion of \textit{best response} and \textit{Nash equilibrium} are
standard: $s_i$ is best response to $s_{-i}$ if for all $s_i' \in
S_i$, $\lambda(s) \succeq_i \lambda(s_i', s_{-i})$; $s$ is a Nash
equilibrium if for all $i \in [n]$, $s_i$ is best response to
$s_{-i}$. Existence of Nash equilibrium and computation of an
equilibrium profile (when it exists) are important questions in the
context of strategic form games.

Given a strategic form game $T$, let $G_T$ denote the improvement
graph associated with $T$ (as defined above). Improvement paths in
$G_T$ correspond to maximal sequence of strategy profiles that arise
by allowing players to make unilateral profitable deviations that
result in improving their choice according to their preference ordering.
We say that a game has the \emph{finite improvement property} (FIP) if
every improvement path in $G_T$ is finite \cite{MS96}. In an
improvement path, if each $k_j$ edge in the sequence is the best
response of agent $k_j$ to $s^{j-1}_{-k_j}$ then is called a best
response improvement path. We can analogously define the finite best
response property (FBRP) if every best response improvement path is
finite. FIP not only guarantees the existence of Nash equilibrium, but
also ensures the stronger property that a decentralised local search
mechanism convergences to a equilibrium outcome. Various natural
classes of resource allocation games like congestion games
\cite{Ros73}, fair cost sharing games and restrictions of polymatrix
games \cite{coord} are known to have the FIP.

A weakening of FIP was proposed by Young \cite{You93} which insists on
the existence of a finite improvement path starting from any initial
strategy profile. Classes of strategic form games that satisfy this
property are called weakly acyclic games. Note that weak acyclicity
ensures that a randomised local search procedure almost surely
convergence to an equilibrium outcome \cite{MardenAS07}. Examples of
classes of games which has this property include congestion games with
player specific payoff functions \cite{Mil96}, certain internet
routing games \cite{ES11} and network creation games \cite{KL13}.

As we can see, the improvement graph presents a data structure for
analysing normal form games. It captures the epistemic reasoning 
underlying player choices: if I were to consider a particular 
profile of choices by all of us, I would rather choose another
strategy to improve by my payoff; in that case, agent $j$ would
revise her choice; and so on, unless we reach a profile from where
none of us has any reason to deviate. Such reasoning is closely
related to {\em pre-play negotiations} studied by game theorists.


\subsection{Voting systems}
Consider an electorate consisting of a set $[n] = \{1, \ldots, n\}$ of
$n$ voters and a set $C$ of $m$
candidates. Let $\mathfrak{R}$ be a voting rule that considers the
preference of each voter over the candidates and chooses a subset of
winning candidates of size $k$ (since $k$ among $m$ candidates have to
be elected).
The strategy sets for all voters are the same $S = \mathfrak{L}(C)=\{
\pi | \pi \text{ is a permutation of } C\}$. The outcome set is $O =
\binom{C}{k}$.  The voting rule $\mathfrak{R}: S^{n} \mapsto O$
specifies which candidates win given the complete preferences of all
voters. We assume that each voter $i$ has a preference ordering
$\prec_i$ over the outcome set $O$. Thus the voting system can be
given by the tuple $L = (n,m,\prec_1,\ldots,\prec_n, \mathfrak{R})$.
 
The improvement graph $G_L$ associated with $L$ is as before: $G_E =
(V, E)$ where $V = S^n$, the set of strategy profiles of voters; $E
\subseteq (V \times [n] \times V)$ is the improvement relation for
voter $i$, given by: $s \rightarrow_{i} s'$ if $\mathfrak{R}(s')
\succ_i \mathfrak{R}(s)$, $s_i \neq s_i'$ and for all $j \neq i$, $s_j
= s_j'$.

Voting equilibria have been studied by Myerson and Weber \cite{MW93}.  
In general, one speaks of the {\em bandwagon effect} in an election if
voters become more inclined to vote for a given candidate as her standing 
in pre-election polls improve, or the {\em underdog effect}, if
voters become less inclined to vote for a candidate as her standing improves.
Myerson and Weber suggest that equilibrium arises when the voters, acting in 
accordance with both their preferences for the candidates and their perceptions 
of the relative chances of candidates in contention for victory, generate an 
election result that justifies their perceptions. Note that the improvement
path again gives us the possibility of `interaction' arising from voter 
preferences, and we can analyse this in the context of specific voting rules.

Given that agents may have incentive to strategically misreport their
preferences, it is natural to study the convergence dynamics when
voting is modelled as a game. Iterative voting \cite{Lev12,RE12,Mei17}
is a formalism that is useful to analyse the strategic dynamics when
at each turn a voter is allowed to alter her vote based on the current
outcome until it converges to an outcome from which no voter wants to
deviate.
In general, the outcome of iterative voting may depend on the order of
voters' changes. Again, voters act myopically, without knowing the
others' preferences. This dynamics is again reflected by the
improvement path as discussed here and sink nodes correspond to 
Nash equilibria. Thus given a voting rule, it is natural to ask
what equilibria are reachable from a given vote profile.

\subsection{Allocation of indivisible goods}
An important problem often studied in economics and computer science
is the allocation of resources among rational agents. This problem is
fundamental and has practical implications in various applications
including college admissions, organ exchange and spectrum
assignment. In this paper, we consider the setting where there are
$[n]$ agents and a set $A=\{a_1,\ldots,a_m\}$ of $m$ indivisible
items. An allocation $\pi: N \to 2^{A}$ such that $\cup_{i \in [n]}
\pi(i)=A$ and for all $i,j \in [n]$, $i \neq j$, $\pi(i) \cap \pi(j) =
\emptyset$. In the most general setting, each agent $i$ has a
preference ordering $\prec_i$ over the allocations. Thus an instance
of an allocation problem can be specified as a tuple $H = ([n], A,
\{\prec_i\}_{i \in [n]})$. Let $\Pi$ denote the set of all
allocations. In this setting, each allocation $\pi$ can be viewed as
defining an outcome and agents have a preference ordering over such
outcomes. In a typical allocation problem, it is often assumed that
the preference ordering for each agent $i$ depends only on the bundle
of items assigned to agent $i$. A special case of the above setting is
when $n=m$ (i.e.\ the number of agents and the items are the same) and
$\pi$ is required to be a bijection. An instance of such an allocation
problem $A$ along with an initial allocation $\pi_0$ defines the well
known Shapley-Scarf housing market \cite{SS74}. If agents are allowed
to exchange items with each other, stability of allocation is a very
natural notion to consider. Core stable outcomes are defined as
allocations in which no group of agents have an incentive to exchange
their items as part of an internal redistribution within the
coalition. The improvement graph structure can capture the dynamics
involved in such a sequence of item exchange in a natural manner. The
associated improvement graph can be defined as $G_H=(V,E)$ where
$V=\Pi$. Since the deviation involves exchange of goods among a subset
of players (rather than a unilateral deviation by a single player),
the edge relation is indexed with a subset $u \subseteq [n]$. That is,
for $\pi, \pi' \in \Pi$ and $u \subseteq [n]$, we have $\pi \to_u
\pi'$ if for all $i \in u$, $\pi(i) \prec_i \pi'(i)$, $\pi(i) \neq
\pi'(i)$ and for all $j \notin u$, $\pi(j)=\pi'(j)$.

A finite path in the improvement graph corresponds to a finite
sequence of exchanges that converge to a stable outcome. An important
question is whether stable allocations always exist and whether a
finite sequence of exchanges can converge to such an allocation. For
the housing market, it is known that a simple and efficient procedure
often termed as Gale’s Top Trading Cycle, can compute an allocation
that is core stable. The allocation constructed in this manner also
satisfies desirable properties like strategy-proofness and Pareto
optimality.

The question of whether decentralised swap dynamics converges to
stable allocation has been studied in various related models of
resource allocation. \cite{maudetSwapPower15} analyses optimality in
the setting where pairs of agents exchange items or services. The
influence of neighbourhood structures in terms of item exchange and
its influence on convergence to stable and optimal allocations is
studied in \cite{gourves2017object,CEM17}. Exchange dynamics
with restricted preferences orderings are considered in
\cite{todo15exchange,todocorecomplexity15}.

Note that the improvement graph here is different from the ones we
discussed earlier in a crucial sense. When agents in $u$ swap goods,
the allocation for other players outside $[u]$ is unaffected. If each
agent's preference ordering depends only on the valuation of the bundle
of items that the agent is allocated then their satisfaction is
unchanged. However, agent $1$ may swap goods with $2$ and then use
some the goods acquired to make a swap deal with $3$ thus leading to
interesting causal chains. In effect the entire space of allocations
may be tentatively explored by the agents and the interesting question
is whether they settle to an `equilibrium'. Again the improvement path
offers interesting dynamics, and we can ask whether from any
allocation, we can reach one where no agent wishes to make a swap.  In
the case where the preference orderings depend on externalities
\cite{BPZ13,GSS18} the situation is similar to what we observed in the
context of strategic form games. Apart from stability, notions of
fairness like envy-freeness, propotionality and maximin share
guarantee are also well studied in the context of allocation of
indivisible items \cite{Bud11,BCM16}. Analysis of the improvement
graph is also useful in the context of fairness notions. Existence of
a finite improvement path terminating in a fair allocation would
indicate the possibility of convergence to a fair allocation in terms
of an exchange dynamics.

\subsection{Remark}
We have set up an improvement graph from three different models used in the
mathematical social sciences. This is of course so that a uniform logical
language can be used to specify properties of all these. One natural 
alternative to consider would be to translate all the models into one,
games being the natural choice, and then induce the improvement graph 
over the defined model. This is certainly possible, but in general this
can lead to an increase in the size of the graph. Moreover since we hope to
use MLFPC not only to unify these contexts but also differentiate them
(in terms of logical resources needed), such a reduction would not be
helpful.

\section{Monadic fixed-point logic with counting}
In this section we present the monadic fixed-point and counting 
extension of first order logic interpreted over improvement graphs. 
The use of the fixed-point extension is motivated by the fact that
we wish to express properties like acyclicity of the graph, which
is not first-order expressible. As we will see below, we need the
fixed-point quantifier to range only over collections of nodes, and
hence monadic fixed-point quantifiers suffice for our purpose. The
counting extension helps us to count nodes in a subgraph, or along
a path; this helps us express notions like fairness of schedules,
which is of relevance in specifying improvement dynamics.

An alternative formalism would be the transitive closure extension of
first order logic. But as Grohe has shown \cite{grohePhd}, monadic lfp
logic is strictly less expressive than transitive closure logic, and
hence we prefer a minimal extension of first order logic that serves
our purposes. Note that the counting extension does not add
expressivess but only succinctness. This is of use when we discuss
concurrent deviation by a subset of players.

\subsection{MLFPC Syntax}
Let $\sigma$ be a first order relational vocabulary. Let 
$(S_i)_{i \in \mathbb{N}}$ be a sequence of monadic relation symbols, 
such that for each $i$, $S_i\ \not\in\ \sigma$. These are the second order
fixed-point variables of the logic.

The set of all MLFPC formulas, $\Phi_{\mathit{MLFPC}}$, can be defined inductively as follows: 
\begin{itemize}
\item Every FO-atomic formula $\alpha$ over $\sigma \cup \{S_1,S_2, \dots\}$ 
is an MLFPC formula.

$\fv^{1}(\alpha)=$ the set of first order free variables in $\alpha$. 
$\fv^{2}(\alpha)=$ the set of all relation symbols $S_i$ occurring in $\alpha$;
$\fv(\alpha)=\fv^{1}(\alpha)\ \cup\ \fv^{2}(\alpha)$. 

\item If $\alpha,\beta$ are MLFPC formulas then so are, $\sim\alpha$, 
$\alpha \wedge \beta$ and  $\alpha \vee \beta$. 

$\fv(\sim\alpha)=\fv(\alpha)$ and 
$\fv(\alpha \wedge \beta)=\fv(\alpha \vee \beta)= \fv(\alpha) \cup \fv(\beta)$.
 
\item If $\alpha$ is a MLFPC formula, $x \in \fv^{1}(\alpha)$ and 
$k \in \mathbb{N}$, then so are $\exists x \alpha$, $\forall x \alpha$ and 
$C_{x} \alpha \leq k$. 
  
$\fv(\exists x \alpha)=\fv(\forall x \alpha)=\fv(C_{x} \alpha \leq k)
= \fv(\alpha)\setminus \{x\}$. 
  
\item If $\alpha$ is an MLFPC formula, $S_i \in \fv^{2}(\alpha)$, 
$x \in \fv^{1}(\alpha)$, and $u \not\in  \fv^{1}(\alpha)$ and 
$S_i$ occurs positively in $\alpha$, then $[\lfp_{S_{i},x}\ \alpha](u)$ is 
an MLFPC formula. 

$\fv([\lfp_{S_{i},x}\ \alpha](u))=\fv(\alpha)\setminus \{S_i,x\} \ \cup \{u\}$. 
   
\end{itemize}

The restriction to positive second order variables in the lfp operator
is essential to provide an effective semantics to the logic. It is a
standard way of ensuring monotonicity, given that we do not have an
effective procedure to test whether a given first-order formula is
monotone on the class of finite $\sigma$-structures
\cite{libkin2013elements}.

It should be noted that the use of positive second order variables in no
way restricts us to contexts where equilibria are guaranteed to exist.
Equilibria are given by graph properties, and these variables allow us
to collect sets of vertices monotonically.

\subsection{MLFPC Semantics}
To interpret formulas, we extend $\sigma$-structures with interpretations for 
the free first order and second order variables (the latter from the given 
sequence $(S_i)_{i \in \mathbb{N}}$). Let $\mathfrak{A} $ be a $\sigma$-structure, which has domain A.
 The semantics for the first order formulas of MLFPC are standard. The semantics for the count and the fixed 
point operator are given below.

$$ \mathfrak{A} \models\ C_x \phi(x,\vec{y}) < k \iff\ | \{a \in A |\ \mathfrak{A},x \mapsto a\models \phi \} | < k$$ 

$$\mathfrak{A},u \mapsto a \models [\lfp_{S_i,x} \alpha](u)\ \iff\  a \in \text{lfp}(f_{\alpha})$$ where let, $[\lfp_{X,x} \alpha](u)=\beta$, $\fv(\beta)=\{x,y_1,\dots,y_n,S_1,\dots,S_m\}$, $f_{\alpha}$ is an operator such that, $f_{\alpha}\ :\ \wp(A) \mapsto \wp(A)$ , $B \subseteq A$, $f_{\alpha}(B)=\{a\in A\ |\ \mathfrak{A},S_i \mapsto B, x \mapsto a\ \models\ \beta \}$

The lfp quantifier induces an operator on the powerset of elements on the 
structure ordered by inclusion. The positivity restriction ensures that the
operator is monotone and hence least fixed-points exist.

\subsection{Properties}
\label{sec:prop}
Since the models of interest are improvement graphs, first order variables
range over nodes in the graph, monadic second order variables range over
subsets of nodes and the vocabulary consists of binary relations $E_u$,
where $u \subseteq [n]$. When $|u| = 1$ and $u = \{i\}$, we will simply
write the relation as $E_i$. We write formulas of the form $E(x,y)$
to denote $\bigvee \limits_{i\ \in\ [n]} E_i(x,y)$. 

We now write special formulas that will be of interest in the sequel.

\begin{itemize}
\item $\mathit{sink}(x) = \forall y. \sim E(x,y)$
\item $\mathit{trap}(S,x) = \forall y. (E(y,x) \implies S(y))$ 
\item $\mathit{acyclic} = \forall u. [\lfp_{S,x}\ \mathit{trap}](u)$
\item $\mathit{reach}(S,x) = \mathit{sink}(x) \lor \exists y (E(y,x) \land S(y))$ 
\item $\mathit{weakly-acyclic} = \forall u. [\lfp_{S,x}\ \mathit{reach}](u)$
\end{itemize}

Now consider the formulas interpreted over improvement graphs of normal 
form games. $sink$ refers to the set of sink nodes, and these are
exactly the Nash equilibria of the associated game. The sentence
$acyclic$ is true exactly when the improvement graph is acyclic
and hence such games have the finite improvement property (as every
improvement path is finite). To see that the sentence captures
acyclicity, note the action of the lfp operator: at the zeroth 
iteration, we get all nodes with in-degree $0$; we then get all
nodes which have incoming edges from nodes whose in-degree is $0$;
and so on. Eventually it collects all nodes through which no path
leads to a cycle. Since the sentence applies to every node, we 
infer that the graph does not contain any cycle.

For weak acyclicity, we require that there exists a finite improvement
path starting from  every node. Again the lfp operator picks up sink
nodes at the zeroth iteration, then all nodes that have a sink node
as successor, and so on. Eventually it collects all nodes that 
start finite improvement paths. The sentence asserts that every node
has this property.

Note that these formulas capture equilibria, FIP and weak acyclicity
on improvement graphs {\em per se}, irrespective of whether they 
arise from normal form games, voting systems, or resource allocations.
But there were some differences in the way these graphs were generated,
and we turn our attention to these differences.

In improvement graphs of resource allocation systems, we defined the
edge relation to be $s \rightarrow_u s'$, where $u \subseteq [n]$, 
if $s[j] = s'[j]$ for $j \not\in u$, $\bigcup_{i \in u} s[i] = 
\bigcup_{i \in u} s'[i]$ and for all $i \in u$, $s'[i] \succ_i s[i]$. 
Note that an analogous definition in the case of normal form games
leads us naturally to a {\em concurrent setting}, and we get what
are called $k$-equilibria in {\em coordination games}.

Let $s$ and $s'$ be strategy profiles in a normal form game, $u
\subseteq [n]$. Define $s \rightarrow_u s'$ when $s[j] = s'[j]$ for $j
\not\in u$, and for all $i \in u$, $s'[i] \succ_i s[i]$. That is, with
the choices of the other agents fixed, the coalition of agents in $u$
can coordinate their choices and deviate to get a better outcome. Thus
$k$-equilibria are nodes from which no coalition of at most $k$ agents
can profitably deviate.  We can then define a coalitional $k$
improvement path, where at each step a coalition of at most $k$ agents
deviate, which leads us further to a $k$-FIP. \cite{WAG} shows that a
class of {\em uniform} coordination games has this property.

\begin{itemize}
\item $\mathit{sink}_u(x) = \forall y. \sim E_u(x,y) $
\item $\mathit{sink}_k(x) = \bigvee\limits_{u \subseteq [n], |u| \leq k} \mathit{sink}_u(x) $
\item $k-\mathit{edge}(x,y) = \bigvee\limits_{u \subseteq [n], |u| \leq k} E_u(x,y) $
\item $k-\mathit{trap}(S,x) = \forall y. (k-\mathit{edge}(y,x) \implies S(y)) $ 
\item $k-\mathit{FIP} = \forall u. [\lfp_{S,x}\ k-\mathit{trap}](u) $
\end{itemize}

Note that the disjunctions are large, exponential in $k$. Since we 
have a counting operator, we could add further structure to nodes, 
prising out the individual strategies of players and then use the 
counting quantifier over these to get a succinct formula, linear 
in $k$. 

Observe that these formulas apply to coalitions of $k$-agents performing
swaps in resource allocation systems and we thus uniformly transfer 
notions from the concurrent setting of games to resource allocation 
systems as well. Proceeding further, we get a similar notion for
voting systems, where again a subset of voters can get together 
and agree on revising their expressed preferences, thus leading 
us to coalitional improvement paths and coalitional FIP.

In general, we might want to specify reachability of a set of
distinguished nodes satisfying some property. For instance, in
the context of allocation systems, we are interested in nodes
that are {\em envy free}. An agent $i$ envies an agent $j$ at 
node $x$ if there exists a node  $y$ such that $y \succ_i x$
and the allocation for $i$ at $y$ is the same as the allocation
for $j$ at $x$. A node is envy-free if no player envies another
at that node. We might then want to assert that an envy free
node is reachable from any node. Note that we only need to
enrich the first order vocabulary to speak of $x[i], y[j]$ etc
to express envy-freeness, and the lfp operator is sufficient
to specify reachability of such nodes.

\begin{itemize}
\item $\mathit{reach}_\phi(S,x) = \phi(x) \lor \exists y (E(y,x) \land S(y)) $
\item $\phi-\mathit{reachable} = \forall u. [\lfp_{S,x}\ \mathit{reach}](u)$
\end{itemize}

In the context of voting systems, a similar strengthening of the
syntax of atomic formulas can lead to amusing specification of
reachability of singular profiles in which all voters have identical
first and last preference but disagree on all candidates ranked by
them in between. 

Further, the counting quantifier can give us interesting relaxations
of improvement dynamics. For instance consider the following 
specifications:

\begin{itemize}
\item $\mathit{reach}(S,x) = \mathit{sink}(x) \lor \exists y (E(y,x) \land S(y))$ 
\item $\mathit{path}-\mathit{count} = C_u ([\lfp_{S,x}\ \mathit{reach}](u)) < 5 $
\end{itemize}

This specifies that at most $5$ nodes have finite improvement
paths originating from them. Clearly, such specifications are
of greater interest in voting systems than in the others. Here
the lfp operator is in the scope of the counting quantifier. 
In the following specification, we have them the other way about.

\begin{itemize}
\item $\mathit{count}-\mathit{trap}(S,x) = C_y E(y,x) < k \implies (\forall z. (E(z,x) \implies S(z))) $
\item $\mathit{special} = \forall u. [\lfp_{S,x}\ \mathit{count}-\mathit{trap}](u)$
\end{itemize}

Consider the housing market model mentioned in the previous
section. This framework corresponds to a resource allocation problem
where a valid allocation corresponds to a bijection between agents and
resources and each agent is initially assigned an item. Stability is
an important solution concept in this setting. The dynamics in which
agents start with the initial allocation and repeatedly exchange items
provided it is a profitable move for all the agents involved in the
exchange is quite natural. The obvious question is whether this
process converges to a stable outcome. If we model the associated
improvement dynamics graph as exchange between pairs of agents then
each action corresponds to resolving a \textit{blocking pair}. The
sentence \textit{sink} then refers to the existence of 2-stable
outcomes and \textit{acyclic} asserts that every sequence of blocking
pair leads to a 2-stable outcome. If we interpret the improvement
dynamics graph as capturing exchange of items within a coalition of
agents then these sentences correspond to existence and convergence to
core-stable outcome. For the housing allocation problem, the existence
of core-stable outcome is guaranteed by the Top Trading Cycle
algorithm. 
There are variants of the housing
allocation problem where stable outcomes are not always guaranteed to
exists (for instance, in the presence of externalities). In this case,
by using the counting quantifier, we can make interesting assertions
like the number of nodes that have finite improvement paths
originating from them.

\section{Model Checking Algorithm}
In this section we discuss the model checking problem for MLFPC. Given
$\phi \in \Phi_{\mathit{MLFPC}}$ and a $\sigma$-structure
$\mathfrak{A}$, the model checking problem is to check if
$\mathfrak{A} \models \phi$. We show that the logic admits an
efficient model checking procedure (Algorithm 1). The main
idea is to use the model checking algorithm for FO and modify it
accordingly for the newly introduced quantifiers of the count and the
least fixed point.

\begin{tcolorbox}[width=5.5in,breakable, title= Algorithm 1 for Checking of $\mathfrak{A} \models \phi$]
\textbf{Input} : $\mathfrak{A}, \phi $ \\
\textbf{Output} : Sol $\subseteq A^{\fv^{1}(\phi) }$ \\
$\text{Sub}_{\phi}$ = $\{$ subformulas of $\phi \}$ /* ordered via subformula ordering $ \leq_{s}$ */ \\
Q = $\{ q_\alpha |\ \alpha\ \in\ \text{Sub}_{\phi} \}$ \\
\begin{enumerate}[\indent {}]
\item{\textbf{for}}  $\alpha \in \text{Sub}_{\phi}$
\begin{enumerate}[\indent {}]
  \item{\textbf{switch}} \textit{type of} $\alpha$
   \begin{enumerate}
      \item{\textbf{case}} $\alpha$ is an atomic formula \\  $q_{\alpha}$ is computed by reading off the structure $\mathfrak{A}$
      \item{\textbf{case}} $\alpha = \sim\beta$ \\ $q_{\alpha}$ is just the negation of $q_{\beta}$ /*since $\beta \leq_s \alpha,\ q_{\beta}$ is already computed */
      \item{\textbf{case}} $\alpha =  \beta_1 \wedge \beta_2$ \\
          Let P = $\fv^1(\beta_1) \cap \fv^1(\beta_2)$,
      Q = $\fv^1(\beta_1) \setminus P$,
      R = $\fv^1(\beta_2) \setminus P$ \\
      \textbf{for} $\vec{x} \in A^{|P|}, \vec{y} \in A^{|Q|}, \vec{z} \in A^{|R|} $ \\
      \begin{enumerate}[\indent {}]
      
       \item{\textbf{if}} $q_{\beta_1}(\vec{x},\vec{y})=1 \text{ and } q_{\beta_2}(\vec{x},\vec{z})=1$ \\
       \hspace*{1em} $q_{\alpha}(\vec{x},\vec{y},\vec{z})=1$\\
      \item{\textbf{else}} $q_{\alpha}(\vec{x},\vec{y},\vec{z})=0$
      \end{enumerate}
      \item{\textbf{case}} $\alpha = \exists y \beta$ \\
      \textbf{for} $\vec{x} \in A^{|\fv^1(\alpha)|}, a \in A $ 
      \\
       \hspace*{1em}\textbf{if} $q_{\beta}(\vec{x},a)=1$
       \\ \hspace*{1.3em}$q_{\alpha}(\vec{x})=1$
      
     \item{\textbf{case}} $\alpha = C_{y} \beta \leq k$ \\
     Let count = 0 \\
     \textbf{for} $\vec{x} \in A^{|\fv^1(\alpha)|}$
     \begin{enumerate}[\indent {}]
     
     \item{\textbf{for}} $a \in A$ \\
     \hspace*{1em} \textbf{if} $q_{\beta}(\vec{x},a)=1$ \textbf{then} count = count + 1
     
     \item{\textbf{if}} count $\leq$ k \textbf{then} $q_{\alpha}(\vec{x})=1$
     
     \item count = 0
     \end{enumerate}
     \item{\textbf{case}} $\alpha =  [\lfp_{S_{i},x}\ \beta](u)$ \\ lfp($\alpha$) /*subroutine call */
        
   \end{enumerate}
\end{enumerate}
\end{enumerate}
\end{tcolorbox}

\begin{tcolorbox}[width=5.5in, breakable, title= Algorithm 2 : lfp($\alpha$) - Subroutine call for computing the least fixed point ]
\textbf{Input} : $\alpha =  [\lfp_{S_{i},x}\ \beta](u)$ \\
\textbf{Output} : $q_{\alpha}$ \\
Let $\fv(\beta)=\{x,\vec{y},S_{i},\vec{Y}\}$ , $\fv(\alpha)=\{u,\vec{y},\vec{Y}\}$ \\
\textbf{for} $\vec{a} \in A^{|\vec{y}|}$\\
\begin{enumerate}[\indent {}]
\item iter = $\{\}$, $f^{\vec{a}}_{\beta} = \{\}$
\item \textbf{do}
\begin{enumerate}[\indent {}]
\item iter = $f^{\vec{a}}_{\beta}$
\item /* call to a FO-model checking procedure, where the second order variables are fixed  */
\item $f^{\vec{a}}_{\beta} = \{b \in A\ |\ \mathfrak{A}, S_i \mapsto f^{\vec{a}}_{\beta}, x \mapsto b, y \mapsto \vec{a} \ \models\ \beta \}$ 
\item \textbf{for} $a\in A$ \\
\hspace*{1em} \textbf{if} $a \in f^{\vec{a}}_{\beta}$ \textbf{then} $q_{\alpha}(\vec{a},a)=1$
  
\end{enumerate}

\item{\textbf{while}} $f^{\vec{a}}_{\beta} \not= iter$
\end{enumerate}
\end{tcolorbox}

\begin{theorem}
Given $\phi$ and $\mathfrak{A}$, Algorithm 1 decides if
$\mathfrak{A} \models \phi$ in polynomial time.
\end{theorem}

\begin{proof}
Given $\varphi \in \Phi_{\textit{MLFPC}}$, let $\text{Sub}_{\varphi}$
denote the set of subformulas of $\varphi$ and $\leq_{s}$ denote the
corresponding subformula ordering.
The idea follows the algorithms outlined for model checking procedures
in FO and combines with what is known about the least fixed point. We
maintain a polynomial time computable relational list $Q$ of
polynomial size.
We will basically follow the proof of the first order logic and argue
similarly for the count and least fixed point operators
introduced. Let $A$ be the underlying domain of $\mathfrak{A}$.
For each $\alpha \in \text{Sub}_{\phi}$ let $q_{\alpha} :
A^{\fv^{1}(\alpha)} \mapsto \{0,1\}$ where $q_{\alpha} (\vec{a}) = 0$
if $\mathfrak{A}, \vec{x} \mapsto \vec{a} \not\models \alpha$ and
$q_{\alpha} (\vec{a}) = 1$ otherwise.

\subparagraph{\textit{Induction Hypothesis}.} For all $\alpha \in \text{Sub}_{\phi}$, $Q_{\alpha}$ is of size $|A|^{O(1)}$ and can be computed in time $|A|^{O(1)}$.

\subparagraph{\textit{Base Case}.} When $\alpha$ is an atomic formula, then $Q_{\alpha}$ can be directly computed from $\mathfrak{A}$. Since each of the relations defined in $\mathfrak{A}$ are polynomial in the size of A and it would take a linear pass across the relations expressed to compute $Q_{\alpha}$. Thus, we can conclude that the total time taken is also polynomial in the size of A. 

\subparagraph{\textit{Induction Case}.}

\begin{itemize}
\item{$\alpha = \sim \beta$} By, I.H, since $\beta \leq_s \alpha$, we
  would have already computed $Q_{\beta}$. $Q_{\alpha} =
  A^{\fv^{1}(\alpha)} \setminus Q_{\beta} $ which can be computed by
  an algorithm running in time $|A|^{\fv^{1}(\alpha)}$.

\item {$\alpha = \beta_1 \wedge \beta_2$} By, I.H. we have in
  polynomial time been able to maintain the polynomial sized lists
  $Q_{\beta_{1}}$ and $Q_{\beta_{2}}$. Then the procedure outlined in
  the algorithm would take time $|A|^{\fv^{1}(\alpha)}$ and maintains
  a list of similar size.

\item{$\alpha = \exists x \beta$} By, I.H. we have in polynomial time
  been able to maintain the polynomial sized lists $Q_{\beta}$. The
  algorithm outlined would take $O(A^{\fv^1(\beta)})$ which would also
  happen to be the size of the list maintained by $Q_{\beta}$.

\item {$\alpha = C_{x} \beta \leq k$ } By, I.H. we have in polynomial
  time been able to maintain the polynomial sized lists
  $Q_{\beta}$. The outlined procedure computes $Q_{\alpha}$ in time
  $O(A^{\fv^1(\beta)})$.

\item{ $\alpha = [\lfp_{S_{i},x}\ \beta](u)$} First we note,
  $\fv(\beta)=\{x,\vec{y},S_{i},\vec{Y}\}$,
  $\fv(\alpha)=\{u,\vec{y},\vec{Y}\}$ and
  $|\fv^{1}(\alpha)|=|\fv^{1}(\beta)|$. What essentially happens here
  is that the second order variable $S_i$ gets used to generate an
  inductive relation via the fixed point computation and the newly
  introduced first order variable u is utilised to check validity of
  these formulas. So, unfortunately we cannot directly use $Q_{\beta}$
  to give a polynomial time procedure to generate $Q_{\alpha}$. We
  would rather use the fact that the least fixed point computation is
  a polynomial time computation even in the case of the first order
  logic with the count introduced. 

If we look at the inductive procedure to compute the fixed point we
see that the monadic fixed point operator starts at an empty set and
then converges to a subset of A. And at each stage there is an
increment in the number of elements of the output set by at least
1. Therefore this will run in maximum $O(|A|)$ time. Now for the call
to model checking at each stage we notice that since the formulas get fixed values assigned at each model checking call, we actually operate a FO-model checking call,
which we know happen to take $O(A^{|\beta|})$ time. We have a precomputation of $Q_{\beta}$, which we can additionally  make use of to reduce the time in the following manner. 
\\
For each choice of $b \in A$ the entire instance gets fixed and in thus a
linear pass through $|\beta|$ we would get to know whether the choice
of $b$ satisfies the formula or not. Therefore each stage that does the
FO-model checking would take time at most $O(|\beta|)$. Therefore
total time taken by the entire lfp computation part is
$O(A^{|\fv^1(\alpha)|} \times |\beta|)$. We are interested in the data
complexity of our procedure, where we can ignore the size of the formula,
which typically happens to be of lesser size than the model over which
the model checking procedure is held (reflects the practical circumstances). So we can conclude that our procedure
runs in time $O(A^{|\fv^1(\alpha)|})$. 
\end{itemize}
\end{proof}

\noindent {\bf Complexity of Algorithm 1.} The algorithm
iterates over all subformulas in increasing order.
The worst case running time of the lfp procedure for inputs
$\mathfrak{A}=(A, E_i,\ldots)$ and $\phi$ is $O(A^{|\fv^1(\alpha)|}
\cdot |\phi|^{2})$, where $|\fv^1(\alpha)| = \mathit{max}\ \{
|\fv^1({\alpha})| \ |\ \alpha \in \mathit{Sub}_{\phi} \}$. Thus the
running time of Algorithm 1 is $O(A^{|\fv^1(\alpha)|}\cdot
|\phi|^{2})$.
For the specific properties mentioned in section \ref{sec:prop}, note
that the corresponding MLFPC formulas refers to one second order
variable and two first order variables. Thus for all the properties
mentioned in section \ref{sec:prop}, the model checking procedure runs
in time $O(A^2 \cdot |\phi|^{2})$.

In the context of improvement graphs, if there are $n$ agents and at
most $m$ choices for each agent, the size of the associated
improvement graph is $O(m^n)$. Since it is possible to have a compact
representation for certain subclasses of strategic form games, for
instance, polymatrix games \cite{Jan68}, the size of the improvement
graph structure can be exponential in the representation of the game.
Thus the model checking procedure, while polynomial on the size of the
underlying improvement graph, can in principle, be exponential in the
size of the game representation. This observation may not be very
surprising since even for restricted classes of games like 0/1
polymatrix games, checking for the existence of Nash equilibrium is
known to be NP-complete \cite{ASW15}.

\section{Discussion}
We see this paper as a preliminary investigation, hopefully leading
to a descriptive complexity theoretic study of fundamental notions
in games and interaction. It is clear that fixed-point computations
underlie the reasoning in a wide variety of such contexts, and 
logics with least fixed-point operators are natural vehicles of
such reasoning. We expect that this is a minimal language for
improvement dynamics, but with further vocabulary restrictions
that need to be worked out. Proceeding further, we would like to delineate
bounds on the use of logical resources for game theoretic reasoning.
For instance, one natural question is the characterization of
equilibrium dynamics definable with at most one second order
(fixed-point) variable. 

Expressiveness needs to be sharpened from the perspective of models
as well. We would like to characterize the class of improvement
graphs for different subclasses of games, resource allocation
systems and voting rues, considering the wide variety of details
in the literature. This would in general necessaite enriching the
logical language and we wish to consider minimal extensions.

Another important issue is the identification of subclasses that
avoid the navigation of huge improvement graphs. Potential games
provide an interesting subclass and they correspond to some
appropriate allocation rules and forms of voting (under specific
election rules). But these are only specific exemplifying instances,
studying the stucture of formulas and their models will (hopefully)
lead us to many such correspondences.

An important direction is the study of infinite strategy spaces.
Clearly the model checking algorithm needs a finite presentation
of the input but this is possible and it is then interesting to
explore convergence of fixed-point iterations.

\bigskip

\noindent {\bf Acknowledgements.}  We thank the reviewers for their
insightful comments. We thank Anup Basil Mathew for discussions on
improvement dynamics. Sunil Simon was partially supported by grant
MTR/ 2018/ 001244.

\bibliographystyle{eptcs}
\bibliography{references.bib}

\end{document}